\documentclass{llncs}
\pdfoutput=1
\usepackage{graphicx} 
\usepackage{enumitem} 
\usepackage{amssymb}
\usepackage{amsmath}
\usepackage{amsfonts}
\usepackage{mathtools} 

\newcommand\LL[1]{
\overline{#1}
}

\newcommand\M[1]{
\ddddot{#1}
}

\newcommand\Part{
\mathcal{P}(G)
}

\newcommand\PP{
P^{\ast}
}

\newcommand\Ll[1]{
\LL{#1_\ell}
}

\newcommand\Lr[1]{
\LL{#1_r}
}

\newcommand\Sc[1]{
\LL{#1_{c}}
}

\newcommand\Ml[1]{
\M{#1_\ell}
}

\newcommand\Mr[1]{
\M{#1_r}
}

\usepackage{mdframed}
\mdfdefinestyle{alg}{
topline=true,innerleftmargin=10,innerrightmargin=10,
frametitlerule=true,
linecolor=black,
linewidth=0.8pt,
skipabove=\dimexpr\topsep+\ht\strutbox\relax}

\makeatletter
\newenvironment{AlgBox}[3]{ 
  \protected@edef\@currentlabelname{#3}
  \protected@edef\@currentlabel{#2}
  \label{#3}
  \begin{mdframed}[style=alg, frametitle={#1}]
}{
  \end{mdframed}
}
\makeatother

\usepackage{alltt}

\newcommand{\figdir}{figures/}
\graphicspath{{\figdir}}

\newcommand\g[1]{
    \makebox[0pt]{\footnotesize \emph{#1}}
}

\newcommand\s[1]{ 
    \makebox[0pt]{\tiny{#1}}
}

\newcommand\vfigbegin[0]{
    \begin{figure*}[ht]
    \begin{center}
}

\newcommand\vfigend[3]{
    \end{center}
    \caption[#3]{#2}
    \label{#1}
    \end{figure*}
}

\usepackage{xspace}
\newcommand{\dual}{offset graph\xspace}
\newcommand{\Dual}{Offset graph\xspace}

\newcommand{\D}{\ensuremath{\mathcal{O}(G)}\xspace}
\newcommand{\DT}{\ensuremath{\mathcal{T}(G)}\xspace}
\newcommand{\DC}{\ensuremath{\mathcal{C}(G)}\xspace}
\newcommand{\DR}{\ensuremath{\mathcal{R}(G)}\xspace}

\title{Drawing bobbin lace graphs, or, Fundamental \\cycles for a subclass of periodic graphs} 

\author{Therese Biedl
    \and
    Veronika Irvine
\thanks{Research supported by NSERC.  Thank you to Anna Lubiw for helpful input.}
}
 \institute{David R.~Cheriton School of Computer Science, University of Waterloo}

\date{} 

\begin{document}

\maketitle 


\begin{abstract}
In this paper, we study a class of graph drawings that arise from bobbin lace patterns.
The drawings are periodic and require a combinatorial embedding with specific properties which we outline and demonstrate can be verified in linear time. In addition, a lace graph drawing has a topological requirement: it contains a set of non-contractible directed cycles which must be homotopic to $(1,0)$, that is, when drawn on a torus, each cycle wraps once around the minor meridian axis and zero times around the major longitude axis.
We provide an algorithm for finding the two fundamental cycles of a canonical rectangular schema in a supergraph that enforces this topological constraint.
The polygonal schema is then used to produce a straight-line drawing of the lace graph inside a rectangular frame.
We argue that such a polygonal schema always exists for combinatorial embeddings satisfying the conditions of bobbin lace patterns, and that we can therefore create a pattern, given a
graph with a fixed combinatorial embedding of genus one.
\end{abstract}

\section{Introduction}

Bobbin lace is a 500-year-old fibre art-form created by braiding threads together in complex patterns.  See Figure~\ref{fig:lace}.
Bobbin lace can depict landscapes, figures, flowers, as well geometric and abstract designs. Common to all bobbin lace compositions is the use of doubly periodic patterns to fill regions of any shape or size.  It is the study of these periodic patterns that we pursue here.
To create bobbin lace, threads, which are wound around wooden bobbins to facilitate handling, are arranged left to right in linear order $t_1,t_2, \dots t_{2n-1}, t_{2n}$.  The lacemaker selects four consecutive threads, starting at an odd index, and crosses the four threads over and under each other to form an alternating braid.  After one or several crossings are made in this manner, the four threads are set aside and another set of four (not the same four, but possibly using a subset of the original four) is selected, again starting at an odd index, and braided.  Since the selected threads are consecutive starting at an odd index, we can describe the pattern by tracking the movement of pairs of threads rather than individual strands.

\vfigbegin
    \def\svgwidth{\textwidth}
\begingroup%
  \makeatletter%
  \providecommand\color[2][]{%
    \errmessage{(Inkscape) Color is used for the text in Inkscape, but the package 'color.sty' is not loaded}%
    \renewcommand\color[2][]{}%
  }%
  \providecommand\transparent[1]{%
    \errmessage{(Inkscape) Transparency is used (non-zero) for the text in Inkscape, but the package 'transparent.sty' is not loaded}%
    \renewcommand\transparent[1]{}%
  }%
  \providecommand\rotatebox[2]{#2}%
  \ifx\svgwidth\undefined%
    \setlength{\unitlength}{546.32744141bp}%
    \ifx\svgscale\undefined%
      \relax%
    \else%
      \setlength{\unitlength}{\unitlength * \real{\svgscale}}%
    \fi%
  \else%
    \setlength{\unitlength}{\svgwidth}%
  \fi%
  \global\let\svgwidth\undefined%
  \global\let\svgscale\undefined%
  \makeatother%
  \begin{picture}(1,0.20831847)%
    \put(0,0){\includegraphics[width=\unitlength]{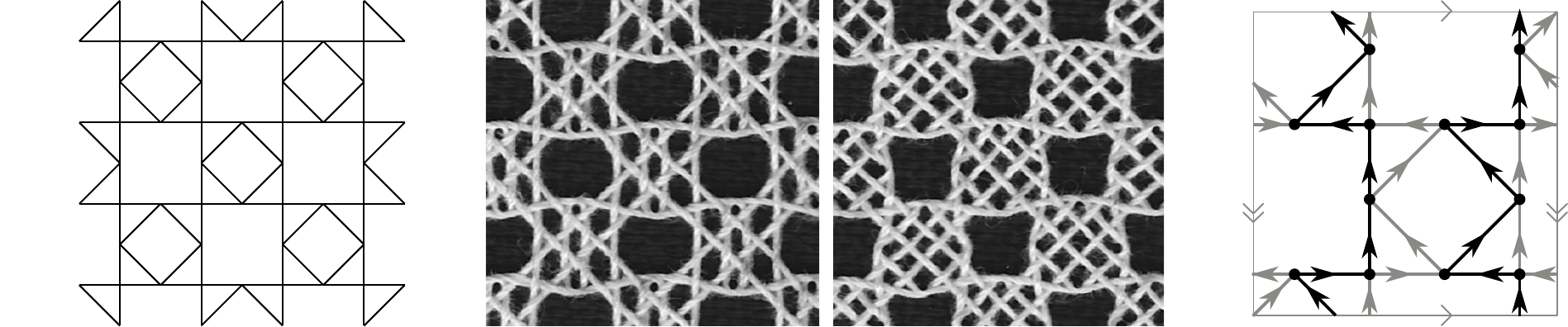}}%
    \put(0.02469437,0.09942235){\color[rgb]{0,0,0}\makebox(0,0)[b]{\smash{\g{(a)}}}}%
    \put(0.28451959,0.09942235){\color[rgb]{0,0,0}\makebox(0,0)[b]{\smash{\g{(b)}}}}%
    \put(0.76762126,0.09942235){\color[rgb]{0,0,0}\makebox(0,0)[b]{\smash{\g{(c)}}}}%
  \end{picture}%
\endgroup%

\vfigend{fig:lace}{Bobbin lace. (a) Pattern, (b) Lace with two different $\zeta(v)$ mappings, (c) Periodic graph drawing.  Four osculating circuits distinguished using black and gray.}{Bobbin lace}

The application of mathematics to the study of fibre arts dates back to the beginnings of computer science.  A survey of various areas, including knitting and weaving, is presented by Belcastro and Yackel \cite{belcastro2008}. Grishanov~et~al. take a deep look into knot theory and its relevance to the topology of textiles \cite{grishanov2009}.
The second author and Ruskey \cite{irvine}
were the first to develop a mathematical model for bobbin lace and express its patterns using graph drawings.
Specifically, a lace pattern can be represented as $(\Gamma(G), \zeta(v))$ where
$G$ is a combinatorial embedding that captures the flow of pairs of threads from one grouping of four to another, $\Gamma(G)$ gives a specific drawing of $G$ which assigns a geometry to the position of the braids, and $\zeta(v)$ is a mapping from each node $v \in V(G)$ to
a mathematical braid word which specifies the over and under crossings performed on the subset of four threads that meet at $v$.  A systematic exploration of different $\zeta(v)$ mappings is straight-forward, although time consuming, and has been undertaken by lacemakers for several traditional patterns; see \cite{irvine} for more details.  Discovering which pairs of threads can be successfully combined, as represented by $\Gamma(G)$, is a much harder task and is therefore the focus of this paper.

The main question investigated in this paper is to decide, for a directed graph
with a fixed rotation system, whether we can construct a straight-line drawing that is a lace pattern.
We argue that
recognizing such graphs can be done in linear time and depends only on their combinatorial structure (some of these results
were reported earlier \cite{irvine}).
Lace patterns are doubly periodic.  As a result, a lace graph can be
drawn on the surface of a torus and lifted to its universal cover, an infinite, periodic, planar graph.
A toroidal embedding can be drawn
as a straight-line planar graph drawing with a rectangular outer face
by choosing a canonical rectangular schema \cite{aleardi2012},\cite{duncan}.  The challenge is to
find two fundamental cycles to serve as the borders of the rectangle under the constraints of the
topological requirements of the lace pattern and restrictions imposed by the drawing algorithm itself.
We show how to find a suitable polygonal schema for any valid combinatorial embedding.

\section{Mathematical model of bobbin lace}
\label{sec:background}
We assume familiarity with graphs and combinatorial embeddings;
see for example \cite{mohar}.  Throughout this paper, $G=(V,E)$
is a directed graph that comes with a rotation system, i.e.,
a clockwise order of edges around each vertex.
We will assume that the rotation system describes a cellular
combinatorial embedding
on an orientable surface
(i.e., the complement of $G$ on the orientable surface is a collection of open topological disks).
The genus of the cellular embedding can be determined from the rotation system by computing the {\em facial walk} consisting of edges and
vertices incident to each face in order while walking around the face.  All embeddings of interest here have an Euler characteristic of 0.

A circuit is a closed path in which a vertex may be visited more than once but no edges are repeated.
A cycle is a closed path in which each vertex and edge is only visited once.
A {\em contractible cycle} is a cycle that can be continuously retracted
to a point.
A canonical rectangular representation of a torus is bounded by a pair of {\em non-contractible} cycles, also referred to as fundamental cycles, that have only one vertex in common, their point of intersection.

\subsection{Conditions on lace pattern graph embeddings}

In a lace pattern graph, every vertex $v$
must have exactly two incoming and two outgoing edges, corresponding
to two pairs of threads that meet at $v$, are braided
together and then separate.  We call such a graph a
$2$-$2$-regular digraph.

Lace patterns are doubly periodic, i.e., they can tile the plane by translation in two non-parallel directions.
One ``tile'' of the repeat can be drawn using a canonical rectangular schema in which
the position of a thread intersecting the horizontal boundaries of the schema  has the same abscissa top and bottom and a thread intersecting the vertical boundaries of the schema has the same ordinate value left and right.  In graph drawing this is called periodic.

We do not want lace created from a pattern to ``fall apart'', i.e., the
graph needs to be suitably connected.  In
graph-theoretical terms, this means $G$ must be a cellular embedding.

For a lace pattern to be workable, there must exist a partial ordering of the vertices
such that, for any directed edge, the braid mapped to the tail vertex is worked before that of the head.
The pattern, when repeated over the infinite plane, cannot contain directed cycles, or equivalently, the associated graph on the torus must be free from contractible directed circuits.

The graph may have loops but they must be non-contractible.   Parallel edges with the same orientation do not represent a change in the four threads being braided and therefore do not appear in a lace graph. Parallel edges with opposite orientation must form a non-contractible directed cycle. In other words, a lace graph may be a multigraph but no loop or parallel edge can be a face.

In summary, the following three conditions are required for the combinatorial embedding of any lace
pattern graph $G$:
\begin{itemize}
\item[C1.] $G$ is a directed $2$-$2$-regular digraph.
\item[C2.] The rotation system of $G$ describes a toroidal cellular embedding in which all facial walks contain at least 3 edges.
\item[C3.] All directed circuits of $G$ are non-contractible.
\end{itemize}

It is easy to check in linear time whether (C1) holds.  To test (C2),
we first compute the facial walks to ascertain the number of faces, and can then determine
whether the  embedding is toroidal since,
by Euler's formula, such an embedding with $n$ vertices and $2n$ edges must have exactly $n$ faces.

To test (C3) in linear time, we take advantage of the regular structure of our digraph via the following lemma:
\begin{lemma}
Presume a $2,2$-regular digraph has a toroidal cellular embedding $G$.
If $G$ has a contractible directed circuit $C$,
then $G$ will have at least one face bounded by a contractible directed circuit.
\label{lem:cycleFace}
\end{lemma}
\begin{proof}
(Sketch)  Arbitrarily declare one face to contain the origin so that ``inside'' and
``outside'' are well-defined for any contractible directed circuit $C$.
Let $O_C$ and $I_C$ be the number of faces outside and inside $C$,
and consider a directed contractible circuit $C$ that maximizes $|O_C - I_C|$.  Prove, by contradiction, that a directed cycle containing an edge in its interior cannot maximize $|O_C - I_C|$ and therefore the inside of $C$ is a face as required.  Details are in the appendix.\qed\end{proof}

It follows that to verify condition (C3), we simply test whether any facial walk is a directed circuit.
Clearly this takes linear time.

There is one more important restriction for a workable bobbin lace pattern. The threads in the periodic pattern are continuous; threads are neither removed (by cutting) nor added (by knotting or weaving in). In other words, to fill a rectangle of fixed width and undetermined height, a fixed set of threads starts at the top of the rectangle and the same set of threads terminates (not necessarily in the same order) at the bottom of the rectangle.  To achieve this, the threads in one repeat of the pattern must not have a net drift to the right or left. See \cite{irvine} for a more detailed discussion of this {\em thread
conservation} property.

To formulate the thread conservation property in a mathematically precise way,
we require additional terminology and some observations resulting from (C1,C2,C3) which we will describe in the next section.

\subsection{Osculating circuits and thread conservation}
\label{sec:osculating}

Fix a digraph $G$ with a combinatorial embedding such that (C1,C2,C3) hold.
There are two ways in which edges can be arranged at a vertex $v$
with $\mathit{indeg}(v)=\mathit{outdeg}(v)=2$:
Either \emph{rotationally alternating} in which edges alternate between
incoming and outgoing directions or \emph{rotationally consecutive} with
edges in the order incoming, incoming, outgoing, outgoing.
Irvine and Ruskey~\cite{irvine}  showed that the following condition is necessary for (C3):
\begin{itemize}
\item[C3$'$.] At all vertices of $G$, the outgoing arcs are rotationally consecutive.
\end{itemize}

Under (C3$'$), we define the \emph{left/right incoming/outgoing} edges of $v$ as follows:
going in clockwise order
around $v$, we encounter first the left outgoing edge, then the right outgoing edge,
then the right incoming edge and finally the left incoming edge.
Consider two edge-disjoint directed circuits $C_1,C_2$ that have a vertex $v$
in common. There are two possible ways in which $C_1$ and $C_2$ can meet at $v$.
In an \emph{osculating intersection} these circuits only touch (``kiss''),
i.e., both incoming and outgoing edges of $C_1$ are on the left side of $v$ and the edges of $C_2$ are on the right side of $v$ (or vice versa).
In contrast, at a \emph{transverse intersection} the circuits truly cross,
i.e., $C_1$ enters from the left side of $v$ and exits from the right, $C_2$ enters from the right and exits from the left (or vice versa).
\begin{lemma}
Presume (C1,C2,C3$'$) hold.  The edges of $G$ can be partitioned into a set $\Part$ of disjoint
directed circuits such that no two circuits in $\Part$ have
a transverse intersection.  Furthermore, this partition is unique and
can be found in linear time.
\end{lemma}
\begin{proof}
Arbitrarily select an edge $e_1$ of $G$ as the start of a circuit $P_1$.  If this edge is
left incoming at its head $v$, then let $e_2$ be the left outgoing edge at $v$, else
let $e_2$ be the right outgoing edge at $v$.  Put differently, $e_1$ and $e_2$
are on the ``same side'' of $v$.
Append $e_2$ to $P_1$ and repeat the operation at the head of $e_2$.
Since the graph is finite, we eventually must close
the circuit $P_1$; this is the first element of the partition $\Part$.
In fact, circuit $P_1$ must exactly finish at edge $e_1$, because for any
edge the rule of ``stay on the same side'' uniquely determines the edge before
and after in the circuit.

Now select some edge $f_1$ of $G$ that was not in $P_1$, and repeat the process
starting at $f_1$.  The new circuit $P_2$ will not contain an edge of $P_1$ by the same ``stay on the same side'' rule. Thus we obtain the next circuit in the partition.  Repeat until all edges
belong to some circuit.  Since there is never a choice about which next edge
to take, the partition is unique.
Each edge is visited exactly once resulting in a linear runtime.
\qed\end{proof}
We call the circuits in $\Part$ the {\em osculating circuits} of $G$,
see also Figure~\ref{fig:lace}(c).
We distinguish two cases of $\Part$ based on whether or not
the osculating circuits are simple directed cycles.  It turns
out that when a circuit visits a vertex twice, $\Part$ has a trivial structure.
\begin{lemma}
Presume (C1-C3) hold.
If some osculating circuit $P\in \Part$ visits a vertex twice,
then $P$ is the only element in $\Part$.
\label{lem:partcircuit}
\end{lemma}
\begin{proof}
Follow $P_1\in\Part$, a non-simple osculating directed circuit, until the first time a vertex $v \in P_1$ is reached for the second time. $P_1$ can be partitioned into a simple directed cycle $C'$ (by taking the part from $v$ to $v$ that we just followed) and a directed circuit $C''$ (the rest of $P_1$). Both $C'$ and $C''$ are incident to $v$.

Fix an arbitrary drawing of $G$ on the torus $\mathcal{T}$.
By (C3) $C'$ is non-contractible,
so cutting $\mathcal{T}$ along $C'$ produces a cylinder.
The cut will split $v$ into two vertices, $v'$ and $v''$,
one on each boundary of the cylinder.

Because of the osculating construction of $P_1$, $C'$ and $C''$ must intersect transversely at $v$ (otherwise, $P_1$ would have terminated the first time it returned to $v$).
Thus $C''$ contains an incident edge at each of the two copies, $v'$ and $v''$, of $v$.  Taking the subpath $D$ of $C''$ between $v'$ and $v''$, we obtain a path that travels on the cylinder from one boundary to the other.  Cutting along $D$ will cut the cylinder into one or more disks.

Now consider some other circuit $P_i\in\Part$, $P_i\neq P_1$.  It cannot have a transverse
crossing with either $C'$ or $C''$, because $P_1$ and $P_i$ are osculating circuits.
Therefore, it intersects neither $C'$ nor $D$ and we can conclude that  $P_i$ resides entirely
within (or on the boundary) of one of the disks of $\mathcal{T}-C'-D$.
But then $P_i$ is a contractible directed circuit, a contradiction.  So
$\Part$ contains no osculating circuits other than $P_1$.
\qed\end{proof}

The osculating partition can contain more than one element, see
Figure~\ref{fig:lace}(c).  It follows from the previous lemma that if $|\Part| > 1$, then all circuits in $\Part$ must be simple directed cycles.

In a transverse intersection under (C3$'$), if $C_1$ uses the left incoming (and right outgoing) edge of $v$ we say that $C_1$ {\em crosses $C_2$ left-to-right} at $v$.
Summing over all shared vertices, we define the {\em algebraic crossing number} of two circuits $C_1$ and $C_2$ to be:
$$\hat{i}(C_1,C_2)=\#\{\text{$C_1$ crosses $C_2$ left-to-right}\}-\#\{\text{$C_1$ crosses $C_2$ right-to-left}\}$$

Finally, using the following well known lemma, we can make a statement about the homotopy class of the osculating circuits in $\Part$:
\begin{lemma}{\cite[p. 209]{stillwell}}
\label{lem:homClass}
Two closed simple curves $C,C'$ have $\hat{i}(C,C')=0$ if and only if they
belong to the same homotopy class.
\end{lemma}
\begin{lemma}If (C1-C3) hold, then all directed cycles $P_i\in\Part$ belong to the same homotopy class.
\label{lem:pHomClass}
\end{lemma}
\begin{proof}
If $\Part$ contains a non-simple circuit then, by Lemma~\ref{lem:partcircuit}, it is the only member of $\Part$ and the claim holds trivially.
Otherwise, the osculating circuits of $\Part$ are all simple closed curves.
None of them intersect transversally, which means that $\hat{i}(P_i,P_j) = 0$ for all pairs of osculating circuits $i\neq j$.
This proves the result by Lemma~\ref{lem:homClass}.
\qed\end{proof}

A {\em (canonical) polygonal schema} for a toroidal graph consists of two
fundamental cycles (called the {\em meridian} $M$ and the {\em longitude} $L$ ) such that
$M$ and $L$ intersect in exactly one point.  Cutting $G$ along
the edges of $M\cup L$ will result in a topological disk.  A circuit $C$
belongs to {\em homotopy class} $(m, \ell)$ (with respect to a fixed polygonal schema)
if $\hat{i}(C,L)=m$ and $\hat{i}(C,M)=\ell$.

With these terms in place, we can now state the thread conservation property via the following constraint:
\begin{itemize}
\item[C4.] There exists a meridian $M$, a longitude $L$
and a partition $\Part$ of edges into osculating directed circuits
such that each circuit in the partition is in the $(1,0)$-homotopy class.
\end{itemize}

\vfigbegin
    \def\svgwidth{\textwidth}
\begingroup%
  \makeatletter%
  \providecommand\color[2][]{%
    \errmessage{(Inkscape) Color is used for the text in Inkscape, but the package 'color.sty' is not loaded}%
    \renewcommand\color[2][]{}%
  }%
  \providecommand\transparent[1]{%
    \errmessage{(Inkscape) Transparency is used (non-zero) for the text in Inkscape, but the package 'transparent.sty' is not loaded}%
    \renewcommand\transparent[1]{}%
  }%
  \providecommand\rotatebox[2]{#2}%
  \ifx\svgwidth\undefined%
    \setlength{\unitlength}{417.56995975bp}%
    \ifx\svgscale\undefined%
      \relax%
    \else%
      \setlength{\unitlength}{\unitlength * \real{\svgscale}}%
    \fi%
  \else%
    \setlength{\unitlength}{\svgwidth}%
  \fi%
  \global\let\svgwidth\undefined%
  \global\let\svgscale\undefined%
  \makeatother%
  \begin{picture}(1,0.20719037)%
    \put(0,0){\includegraphics[width=\unitlength,page=1]{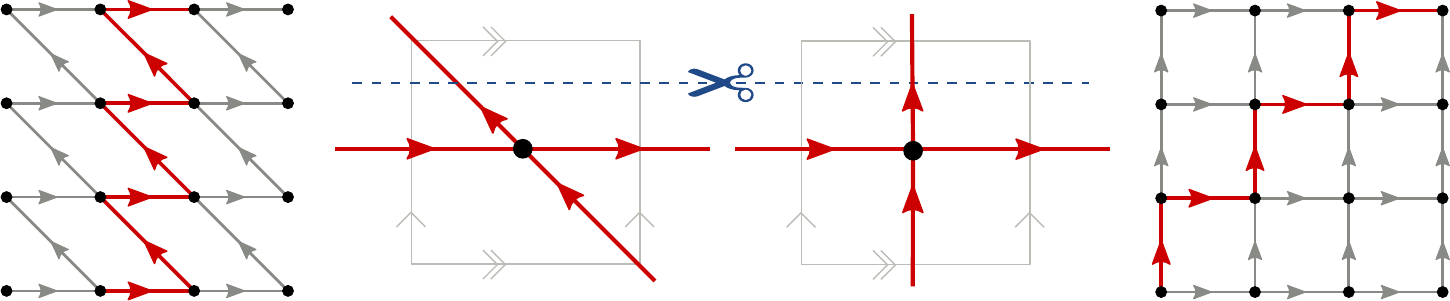}}%
  \end{picture}%
\endgroup%

\vfigend{fig:Dehn}{A Dehn twist changes a valid lace graph on left into an invalid one on right.}{Dehn twist}
The $(1,0)$-homotopy class restriction ensures that at each upward repeat all thread pairs return to the same left-right starting position.
The thread conservation property is impossible to formulate as a condition of the combinatorial embedding
because the homotopy class is affected by how the graph
is drawn on the torus.  In particular, consider Figure~\ref{fig:Dehn}
which shows two drawings of the same graph on the torus differing by a homeomorphism known as a Dehn twist.  Both drawings
have the same combinatorial embedding, yet on the left side of Figure~\ref{fig:Dehn}
the red (bold) osculating circuit returns to its starting point
while on the right side of the figure the osculating circuit drifts to the right.

Clearly, thread conservation demands that we fix more than the combinatorial embedding of the graph.
However, based only on the combinatorial embedding, we can make a statement about the existence of a suitable graph drawing:
\begin{lemma}
Given a digraph $G$ that satisfies (C1-C3), there
exists a drawing of $G$ for which (C4) holds.
\end{lemma}
\begin{proof}
By the Dehn-Lickorish theorem (see e.g.~\cite{farb2011primer})
there exists a homeomorphism that maps any simple, non-contractible cycle to the $(1,0)$ homotopy
class of the torus.  By Lemma~\ref{lem:pHomClass}, such a homeomorphism will map all elements in $\Part$ to the desired homotopy class.
\qed\end{proof}
The main contribution of this paper is a linear time algorithm for finding such a drawing:
\begin{theorem}
Given a digraph $G$ that satisfies (C1-C3), we can draw a lace pattern
in linear time.
The drawing resides in an $O(n^4) \times O(n^4)$-grid.
\label{theo:main}
\end{theorem}
\begin{proof}
In Section~\ref{sec:draw}, we provide an
algorithm for finding a polygonal schema to satisfy (C4)
for any digraph $G$ that satisfies (C1-C3). We then use known algorithms (\cite{duncan}, \cite{aleardi2012}) for straight-line rectangular-frame drawings to create a lace pattern in the required time and space.
\qed\end{proof} 

\section{Finding a polygonal schema}
\label{sec:draw}

In general, finding a polygonal schema with vertex-disjoint interiors
is NP-hard~\cite{duncan}.  However, for the purpose of drawing the
lace-pattern, we do not need to find a polygonal schema within the given graph;
it suffices (and in fact, is preferable) to add vertices and edges
to the graph and find a polygonal schema within the additions.  In this manner, the original
edges of $G$ are not on the schema boundary giving more freedom to where they can be placed.
In this section we describe how to find such a supergraph with $O(n)$ nodes.
\begin{AlgBox}{DrawLacePattern Algorithm}{DrawLaceGraph}{alg:draw}
\begin{enumerate}[label=\Alph*.]
  \item Partition $G$ into a set $\Part$ of osculating circuits and select one directed circuit $\PP$ from the set.
  \label{step:partition}
  \item Create the \dual \D.
  \label{step:dual}
  \item Find a simple cycle $M$ in \D such that $\hat{i}(M,\PP) = 0$
	and $M$ intersects every edge of $G$ at most once.
  \label{step:findM}
  \item Find a simple cycle $L$ in \D such that $\hat{i}(L,\PP) = {\pm}1$,
    $M$ and $L$ intersect exactly once, and $L$ intersects every edge
	of $G$ at most once.
  \label{step:findL}
  \item Use existing torus-drawing techniques to draw $G$ on a rectangle with meridian $M$ and longitude $L$.
  \label{step:canonical}
\end{enumerate}
\end{AlgBox}

All steps are linear or constant time, and step~\ref{step:canonical} will create a drawing of
the required size, proving the theorem.
Step~\ref{step:partition} was explained
in Section~\ref{sec:osculating} already;
all other steps are explained below.

\textbf{Creating the offset graph:}
We first create an \dual $\D$ in which
we will choose a suitable meridian $M$ and longitude $L$.
Roughly speaking, $\D$ is obtained by creating two copies of every
osculating circuit $P$ in $\Part$, which we will represent as $\M{P}$ (used to find $M$) and $\LL{P}$ (used to find $L$).  In each copy, the circuits are separated (they do not share vertices as they do in $G$) and are simple cycles.  The copies lie on top of $G$ introducing crossings which we remove by inserting dummy vertices.  Finally, for the simple cycles of $\LL{P}$, we connect the two halves of each vertex, split in the process of separating osculating paths, by introducing ``crossover-edges''.

\vfigbegin
    \def\svgwidth{\textwidth}
\begingroup%
  \makeatletter%
  \providecommand\color[2][]{%
    \errmessage{(Inkscape) Color is used for the text in Inkscape, but the package 'color.sty' is not loaded}%
    \renewcommand\color[2][]{}%
  }%
  \providecommand\transparent[1]{%
    \errmessage{(Inkscape) Transparency is used (non-zero) for the text in Inkscape, but the package 'transparent.sty' is not loaded}%
    \renewcommand\transparent[1]{}%
  }%
  \providecommand\rotatebox[2]{#2}%
  \ifx\svgwidth\undefined%
    \setlength{\unitlength}{2029.52871094bp}%
    \ifx\svgscale\undefined%
      \relax%
    \else%
      \setlength{\unitlength}{\unitlength * \real{\svgscale}}%
    \fi%
  \else%
    \setlength{\unitlength}{\svgwidth}%
  \fi%
  \global\let\svgwidth\undefined%
  \global\let\svgscale\undefined%
  \makeatother%
  \begin{picture}(1,0.23868252)%
    \put(0,0){\includegraphics[width=\unitlength]{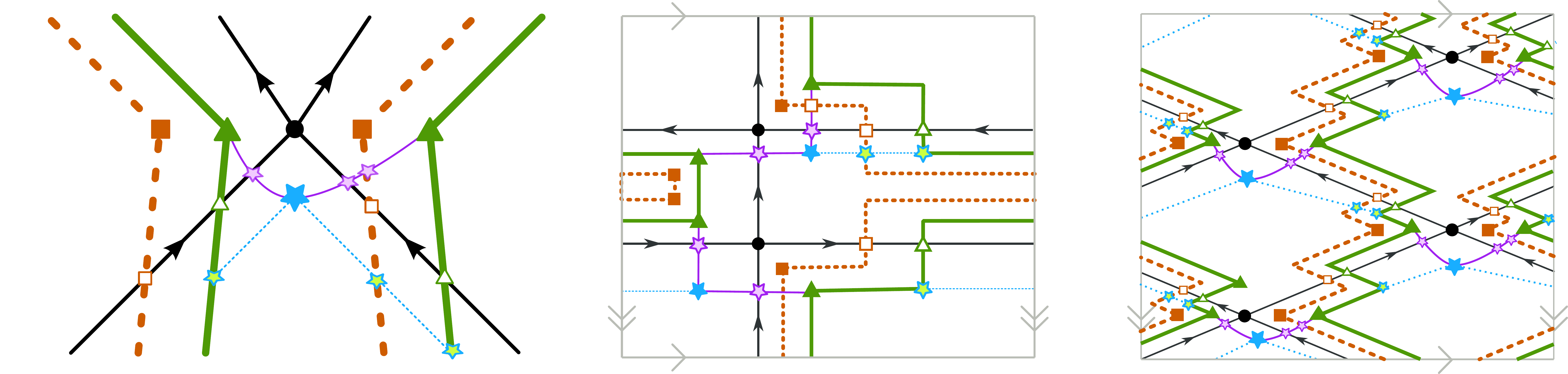}}%
    \put(0.36368074,0.11687801){\color[rgb]{0,0,0}\makebox(0,0)[lb]{\smash{\g{(b)}}}}%
    \put(0.01212085,0.11687801){\color[rgb]{0,0,0}\makebox(0,0)[lb]{\smash{\g{(a)}}}}%
    \put(0.68079013,0.11687801){\color[rgb]{0,0,0}\makebox(0,0)[lb]{\smash{\g{(c)}}}}%
    \put(0.17785812,0.18280759){\color[rgb]{0,0,0}\makebox(0,0)[lb]{\smash{$v$}}}%
    \put(0.05100356,0.16691056){\color[rgb]{0,0,0}\makebox(0,0)[lb]{\smash{\s{$\Ml{v}$}}}}%
    \put(0.11217149,0.16691056){\color[rgb]{0,0,0}\makebox(0,0)[lb]{\smash{\s{$\Ll{v}$}}}}%
    \put(0.25314831,0.16691056){\color[rgb]{0,0,0}\makebox(0,0)[lb]{\smash{\s{$\Mr{v}$}}}}%
    \put(0.31304791,0.16691056){\color[rgb]{0,0,0}\makebox(0,0)[lb]{\smash{\s{$\Lr{v}$}}}}%
    \put(0.18974848,0.07688567){\color[rgb]{0,0,0}\makebox(0,0)[lb]{\smash{\s{$\Sc{v}$}}}}%
  \end{picture}%
\endgroup%

\vfigend{fig:find_l_m}{Create \dual. (a) Close-up near a vertex showing $\M{P}$ (thick, dashed, square, orange), $\LL{P}$ (thick, solid, triangle, green), crossover-edges (thin, solid, purple) and shortcut-edges (thin, dotted, blue). (b) \Dual of non-simple osculating circuit.  (c)  \Dual with multiple simple osculating cycles.}{Create \dual}

Figure~\ref{fig:find_l_m}
illustrates two such toroidal graph embeddings  $\D$, with $O(n)$ vertices and edges, which we will now define formally. The {\em left-face} of a directed edge $v\rightarrow w$ is the face to the left of it while walking from $v$ to $w$.
For each vertex in $\D$ we define $f_{\ell}(v)$ to be the left-face of the left incoming and left outgoing edges of $v$, $f_{r}(v)$ to be the right-face of the right incoming and
right outgoing edges of $v$, and $f_{c}(v)$ to be the face incident to the left incoming and right incoming edges of $v$.
\begin{itemize}
\item Initially, $\D$ contains all vertices and edges of $G$, embedded as in $G$.
\item For every vertex $v\in V(G)$, add four new vertices $\Ml{v},\Ll{v},\Mr{v},\Lr{v}$.
    Place $\Ml{v},\Ll{v}$ in $f_{\ell}(v)$ such that $\Ll{v}$ is closer to the left incoming edge than $\Ml{v}$.
	Place $\Mr{v},\Lr{v}$ in $f_{r}(v)$ such that $\Mr{v}$ is closer to the right incoming edge than $\Lr{v}$.
\item For every edge $e=v\rightarrow w\in E(G)$, we add two new edges $\M{e}$ and $\LL{e}$.
	If $e$ is left outgoing at $v$, then $\M{e}$ starts at $\Ml{v}$ and $\LL{e}$ starts at $\Ll{v}$.
	else $\M{e}$ starts at $\Mr{v}$ and $\LL{e}$ starts at $\Lr{v}$,
	If $e$ is left incoming at $w$, then $\M{e}$ ends at $\Ml{w}$ and $\LL{e}$ ends at $\Ll{w}$,
	else $\M{e}$ ends at $\Mr{w}$ and $\LL{e}$ ends at $\Lr{w}$.
\item Note that edge $e\in E(G)$ may be intersected by its copies
	$\M{e}$ and $\LL{e}$.  This occurs, for example, when $e$ is right outgoing at its tail
	and left incoming at its head.  We remove these crossings (so that we
	again have a toroidal embedding) in the
	standard way by inserting dummy vertices that subdivide $e,\M{e},\LL{e}$.
	(In the following descriptions, we will ignore these dummy-vertices,
	and speak of an edge $e$, even though it has become a
	path with 3 edges.)
\item For every osculating circuit $P$ there are now two circuits $\M{P}$ and $\LL{P}$,
	using for each edge $e\in P$ the corresponding edges $\M{e}$ and $\LL{e}$.
	Note that $\M{P}$ and $\LL{P}$ are simple, even if $P$ is not.  When $P$ visits a
	vertex $v$ twice, once via the incoming left and outgoing left
	edges of $v$ and once via the incoming right and outgoing right edges of $v$,
    the corresponding $\M{P}$ visits $\Ml{v}$ and $\Mr{v}$ respectively and similarly $\LL{P}$ visits $\Ll{v}$ and $\Lr{v}$.
    Due to the order of the copies near $v$, $\M{P}$ and $\LL{P}$ do not cross.
\item Next add the {\em crossover-edge} $\LL{e_v}=(\Ll{v},\Lr{v})$  for each vertex $v\in V$.
	To obtain a toroidal embedding, route
	this edge so that it crosses three edges: the two incoming edges of $G$ at $v$ and the
	edge $\M{e}$ that is incoming to $\Mr{v}$.  These crossings are again
	replaced by dummy-vertices.
    In the crossover-edge insert a vertex $v_c$ in the face $f_c(v)$.
\item For a straight-line drawing, an edge $e$ in $E(G)$ must not cross the rectangular frame twice.
Consider an edge $\LL{e}$ that crosses $e=v\rightarrow w$ where $\LL{e}$ originates inside the face $f_c(w)$,
say $\LL{e}=\Lr{v}\rightarrow \Ll{w}$.
It may happen that the chosen longitude $L$ contains $\LL{e}$ followed by the crossover-edge $\LL{e_w}=(\Ll{w},\Lr{w})$ resulting in a double crossing of $e$ by $L$.
To avoid this situation, we add a {\em shortcut-edge} to $\D$ that connects a point on $\LL{e}$ inside $f_c(w)$ to $w_c$.
Note that any circuit using $\LL{e}\,\LL{e_w}$ acts the same (with respect to algebraic crossing numbers) as a circuit shortened
via the shortcut-edge, since all other crossings are unaffected.
\end{itemize}

\textbf{Finding the meridian:}
In step~\ref{step:partition} we selected an osculating circuit $\PP$.
Define $M$ to be the copy $\M{\PP}$ of circuit $\PP$ in the offset graph.
$M$ is a simple cycle.   It may cross $\PP$ repeatedly but it does
so only by switching back and forth between being left of $\PP$
and right of $\PP$.  In other words, $\hat{i}(\PP,M)=0$ as desired.
Finally $M$ intersects every edge of $G$ at most once by construction of $\M{P}$.

\textbf{Finding the longitude:}
Define $\mathcal{L}(G)$ to be the subgraph of $\D$ taking only the
copies $\LL{P}$ of osculating circuits, the crossover-edges, and the shortcut-edges.
We claim that we can find a suitable longitude in $\mathcal{L}(G)$.
The algorithm is quite simple (find a shortest path within $\mathcal{L}(G)$, with
some edges removed), but arguing that it works is not.

\textbf{Case 1) Circuit $\PP$:}
First consider the easier case in which $\PP$ is not simple and is
therefore, by Lemma~\ref{lem:partcircuit}, the only element in $\Part$.  Let $v$ be a vertex visited twice by $\PP$, hence $\Ll{v}$
and $\Lr{v}$ both belong to $\LL{\PP}$.
Define $L^\ast$ to be a subpath of $\LL{\PP}$ between $\Ll{v}$ and $\Lr{v}$ and
$\hat{L}$ to be $L^\ast$ plus the crossover-edge $(\Lr{v},\Ll{v})$.
We have constructed $\LL{\PP}$ such that it does not intersect $\M{\PP}$.  So $L^\ast$ does
not intersect $M$.  The crossover-edge $(\Lr{v},\Ll{v})$ intersects $M$ exactly once.  Therefore, $\hat{L}$ intersects $M$ exactly once as desired.

$L^\ast$ and $\PP$ may intersect numerous times between $\Ll{v}$ and $\Lr{v}$.
But $\Ll{v}$ is to the left of $\PP$ while $\Lr{v}$ is to the right, so in total $L^\ast$
has one more left-to-right intersection than right-to-left-intersection.
The crossover-edge $(\Lr{v},\Ll{v})$ intersects both incoming edges at $\LL{v}$,
and both of these edges belong to $\PP$.  This adds two more
right-to-left intersections, and so in total $\hat{i}(\hat{L},\PP)=-1$ as desired.

\textbf{Case 2) Simple $\PP$:}
If $\PP$ is simple then we find $L$ using a path that connects ``both sides''
of $\PP$ without crossing $\PP$.  The existence of such a path is non-trivial, and crucially needs
condition (C2), i.e., that the embedding is cellular.  Specifically, we show:
\begin{lemma}
\label{lem:connectingPath}
Presume (C1-C3) hold.
Let $P$ be a directed osculating cycle.  Then there exists a directed walk $W$
in $G$ that starts at a vertex $v\in P$ with a left outgoing edge
$e_1\not\in P$, ends at a vertex $w\in P$ with a right incoming edge
$e_k\not\in P$, and has no transverse intersection or shared edges with $P$.
\end{lemma}
\begin{proof}
(Sketch)  The edge $e_1$ must exist otherwise the directed cycle $P$ would bound a
face, contradicting (C2) or (C3).  One can
now argue that starting from $e_1$ and always taking the left outgoing
edge, we must reach such an edge $e_k$ or find a contradiction
to (C2).  Details are in the appendix.
\qed\end{proof}

Let $Q$ be such a walk in $G$ from $v\in \PP$ to $w\in \PP$, shortened by
eliminating directed cycles (if any) so that it becomes a simple path.
We obtain the longitude
$L$ by ``translating'' $Q$ into $\mathcal{L}(G)$ and adding a subpath of $\LL{\PP}$.

Note that the directed edges of $Q$ need not be rotationally consecutive.
For example, $e_i\in E(Q)$ may be right incoming at its head vertex $x$ while $e_{i+1}~\in~E(Q)$ may be left outgoing at its tail vertex $x$.
If $x\not\in \PP$ then we can can add the crossover-edge $(\Ll{x},\Lr{x})$ to connect $\LL{e_i}$ and $\LL{e_{i+1}}$ without crossing $M=\M{\PP}$.  If $x\in\PP$ but $e_{i}\not\in\PP$ and $e_{i+1}\not\in\PP$ then no crossover-edge is required because $e_{i}$ and $e_{i+1}$ are rotationally consecutive.  It is not possible to have $x\in\PP$ and only one of $e_{i}\in \PP$ or $e_{i+1}\in \PP$ because of the way in which $Q$ is defined.

Formally, let $Q=e_1,\dots,e_k$ be a simple path in $G-E(\PP)$ from $v\in \PP$ to $w\in \PP$  and let $L^{-}=\LL{e_1},\dots,\LL{e_k}$
be a simple path using the corresponding edges in $\mathcal{L}(G)$ and crossover-edges as needed.
$L^{-}$ begins at $\Ll{v}$ since $e_1$ is left outgoing and ends at $\Lr{w}$
since $e_k$ is right incoming. Define $L^{\ast}$ to be the subpath of $\LL{\PP}$ connecting $\Ll{w}$ to $\Lr{v}$. Finally, define $\hat{L}$ to be the simple cycle consisting of $L^{-}$, crossover-edge $(\Lr{w},\Ll{w})$, $L^{\ast}$ and crossover-edge $(\Lr{v},\Ll{v})$.

Note that $L^{-}$ and $L^{\ast}$ never cross the meridian $M=\M{\PP}$.
The only place where $\hat{L}$ can intersect $M$ is at
the two crossover-edges $(\Lr{v},\Ll{v})$ and $(\Lr{w},\Ll{w})$.   We claim that
exactly one of them intersects $\M{\PP}$.  To see this, recall that $\PP$ uses the
right incoming and outgoing edge at $v$ and the left incoming and outgoing
edge at $w$.   The order of vertices in the vicinity of $v$ is $\Ml{v},\Ll{v},v,\Mr{v},\Lr{v}$,
and $\M{\PP}$ uses $\Mr{v}$, so $(\Lr{v},\Ll{v})$ crosses $\M{\PP}$.  On the other hand
the order of vertices in the vicinity of $w$ is $\Ml{w},\Ll{w},w,\Mr{w},\Lr{w}$,
and $\M{\PP}$ uses $\Ml{w}$, so $(\Lr{w},\Ll{w})$ does not cross $\M{\PP}$.    Therefore, $L$ and $M$
cross exactly once as desired.

To see that $\hat{i}(\hat{L},\PP)={\pm}1$, observe that $L^{-}$ has no transverse intersections
with $\PP$ and so contributes nothing.
$L^\ast$ may intersect $\PP$ repeatedly, alternatingly from left to right and
from right to left, however, $L^\ast$ starts on the left side of $\PP$ at $\Ll{w}$ and ends on the right side of $\PP$ at $\Lr{v}$, so it has a net of one left to right crossing.  The crossover-edges
$(\Lr{w},\Ll{w})$ and $(\Lr{v},\Ll{v})$ are both right to left crossings.  Thus $\hat{L}$ has one more crossing from right to left than it has crossings from left to right, yielding $\hat{i}(\hat{L},\PP)=-1$ as desired.

\vfigbegin
    \def\svgwidth{\textwidth}
\begingroup%
  \makeatletter%
  \providecommand\color[2][]{%
    \errmessage{(Inkscape) Color is used for the text in Inkscape, but the package 'color.sty' is not loaded}%
    \renewcommand\color[2][]{}%
  }%
  \providecommand\transparent[1]{%
    \errmessage{(Inkscape) Transparency is used (non-zero) for the text in Inkscape, but the package 'transparent.sty' is not loaded}%
    \renewcommand\transparent[1]{}%
  }%
  \providecommand\rotatebox[2]{#2}%
  \ifx\svgwidth\undefined%
    \setlength{\unitlength}{1686.690625bp}%
    \ifx\svgscale\undefined%
      \relax%
    \else%
      \setlength{\unitlength}{\unitlength * \real{\svgscale}}%
    \fi%
  \else%
    \setlength{\unitlength}{\svgwidth}%
  \fi%
  \global\let\svgwidth\undefined%
  \global\let\svgscale\undefined%
  \makeatother%
  \begin{picture}(1,0.30642567)%
    \put(0,0){\includegraphics[width=\unitlength]{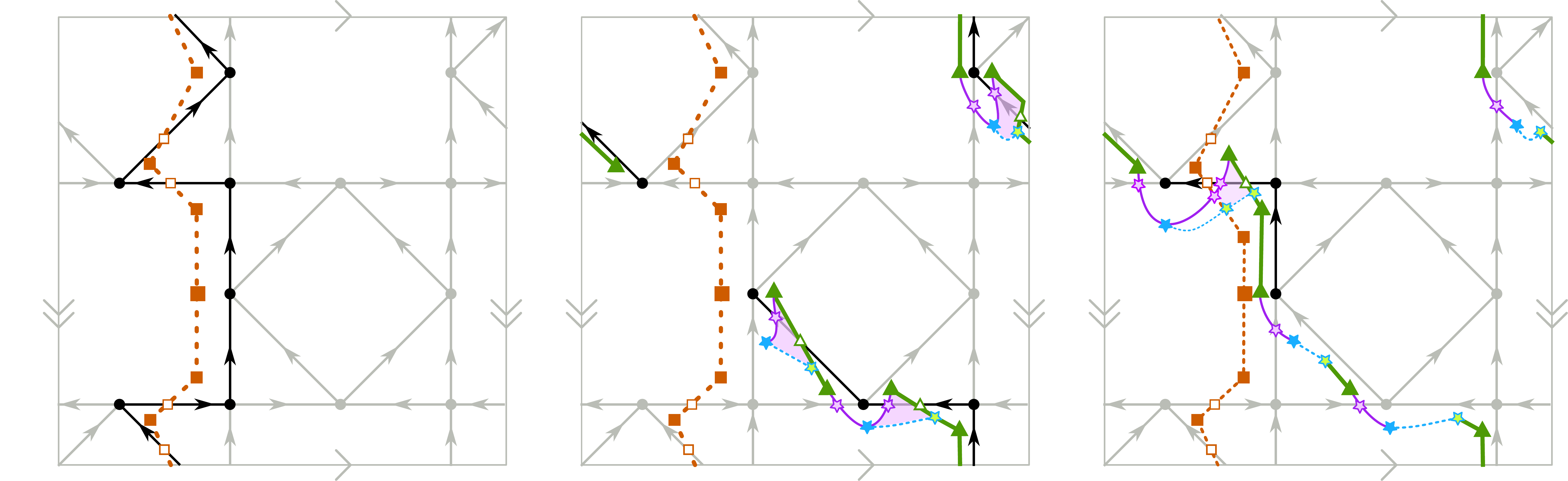}}%
    \put(0.11959695,0.3003865){\color[rgb]{0,0,0}\makebox(0,0)[lb]{\smash{$\PP$}}}%
    \put(0.06968963,0.17409524){\color[rgb]{0,0,0}\makebox(0,0)[lb]{\smash{$v$}}}%
    \put(0.15222333,0.11701865){\color[rgb]{0,0,0}\makebox(0,0)[lb]{\smash{$w$}}}%
    \put(0.08669228,0.3003865){\color[rgb]{0.80784314,0.36078431,0}\makebox(0,0)[lb]{\smash{$M$}}}%
    \put(0.7375655,0.17409524){\color[rgb]{0,0,0}\makebox(0,0)[lb]{\smash{$v$}}}%
    \put(0.8254351,0.11790797){\color[rgb]{0,0,0}\makebox(0,0)[lb]{\smash{$w$}}}%
    \put(0.93938225,0.30033121){\color[rgb]{0.30588235,0.60392157,0.02352941}\makebox(0,0)[lb]{\smash{$L$}}}%
    \put(0.75450887,0.3003865){\color[rgb]{0.80784314,0.36078431,0}\makebox(0,0)[lb]{\smash{$M$}}}%
    \put(0.40407222,0.17409524){\color[rgb]{0,0,0}\makebox(0,0)[lb]{\smash{$v$}}}%
    \put(0.38031462,0.2007747){\color[rgb]{0.30588235,0.60392157,0.02352941}\makebox(0,0)[rb]{\smash{\s{$\LL{v_\ell}$}}}}%
    \put(0.38994887,0.21590153){\color[rgb]{0.05490196,0.1372549,0.18039216}\makebox(0,0)[lb]{\smash{\s{$e_1$}}}}%
    \put(0.47415551,0.1312477){\color[rgb]{0,0,0}\makebox(0,0)[lb]{\smash{$w$}}}%
    \put(0.55516867,0.07666767){\color[rgb]{0.05490196,0.1372549,0.18039216}\makebox(0,0)[rb]{\smash{\s{$e_k$}}}}%
    \put(0.50234918,0.11790797){\color[rgb]{0.30588235,0.60392157,0.02352941}\makebox(0,0)[lb]{\smash{\s{$\LL{w_r}$}}}}%
    \put(0.42101559,0.3003865){\color[rgb]{0.80784314,0.36078431,0}\makebox(0,0)[lb]{\smash{$M$}}}%
    \put(0.01367303,0.15004835){\color[rgb]{0,0,0}\makebox(0,0)[lb]{\smash{\g{(a)}}}}%
    \put(0.34706087,0.15004835){\color[rgb]{0,0,0}\makebox(0,0)[lb]{\smash{\g{(b)}}}}%
    \put(0.68055409,0.15004835){\color[rgb]{0,0,0}\makebox(0,0)[lb]{\smash{\g{(c)}}}}%
  \end{picture}%
\endgroup%

\vfigend{fig:find_L}{(a) $M$ is a shifted copy of $\PP$.  (b) $L$ exits from the left side of $M$ at $e_{1}$ and returns on the right side of $M$ at $e_{k}$.  Loops where shortcuts are required are shaded purple.  (c) Connect $e_1$ to $e_k$ following $\PP$.}{Find $L$}

\textbf{Applying shortcuts:}
As desired, the simple cycle $\hat{L}$ intersects $M$ once
and has $\hat{i}(\hat{L},\PP)=\pm 1$.  However,
$\hat{L}$ may intersect an edge $e$ of $G$ repeatedly. This happens
only when $\LL{e}$ transversely crosses $e$, and a crossover-edge at the tail
of $\LL{e}$ follows immediately after.  Let $L$ be the simple cycle obtained from
$\hat{L}$ by substituting the shortcut-edge at any such edge $e$.  This has the same
algebraic crossing number with $\PP$, still crosses $M$ exactly once, and
is a suitable longitude.

We finally remark that the simplest way to find $L$ is as follows.
Start with graph $\mathcal{L}(G)$ and remove all dummy-vertices that lie on $\M{\PP}$.  This effectively
disallows any path in $\mathcal{L}(G)$ that crosses $M$.  Pick any vertex $v\in \PP$
and, in what remains of $\mathcal{L}(G)$, find a shortest path
$L^{-}$ from $\Ll{v}$ to $\Sc{v}$ (the vertex within
the crossover-edge $(\Lr{v},\Ll{v})$). Finally add the edge $(\Sc{v},\Ll{v})$ to close the cycle
into a suitable longitude $L$.

\textbf{Drawing the graph:}
Let $\DT$ be the graph obtained from $\D$ by removing all edges that do not belong
to $G, M$ or $L$ and smoothing any degree-2 vertices. $\DT$ has a natural embedding on a torus with a well defined meridian and longitude.
Cut $\DT$ along $L$ to form graph $\DC$ which
has a natural embedding on a cylinder with top and bottom boundaries $L_t$ and $L_b$.
Cut $\DC$ along $M$ to obtain a planar graph
$\DR$ with a natural embedding on a rectangle whose top/right/bottom/left
sides are formed by the four boundary-paths $L_t/M_r/L_b/M_{\ell}$.
Figure~\ref{fig:find_l_m_2} in the appendix illustrates this process.

In our construction, we have been careful to ensure that boundary-cycles do not cross an edge of $G$ more than once.
As a consequence, we can now apply known algorithms for straight-line rectangular-frame drawings to $\DR$.  The first such
algorithm was given by Duncan~et~al.~\cite{duncan} and creates
a drawing in an $O(n)\times O(n^2)$-grid in linear time.  Unfortunately, the
algorithm does not necessarily produce a periodic drawing.
We therefore turn to the drawing algorithm of
Aleardi~et~al.~\cite{aleardi2014}, a modification of the algorithm of Duncan~et~al.
which ensures periodic drawings.  This is still achieved in linear time
but at the cost of increasing the grid-size to $O(n^4)\times O(n^4)$-grid.
This proves Theorem~\ref{theo:main}.

\section{Discussion and open problems}

In this paper we provide an algorithm to generate a lace pattern given
a directed graph and a fixed combinatorial embedding that satisfy the
conditions (C1-C3).
Our main challenge was to produce a graph drawing such that the threads
in the pattern weave back and forth within a fixed width (follow a directed circuit that wraps once around the minor meridian axis and zero times around the major longitude axis).  It turns out that if (C1-C3) are satisfied
we can easily define a supergraph of linear size supporting a suitable
rectangular schema and extract the schema in linear time.  With care, we can sure that neither the meridian nor longitude of the schema intersects an edge of the graph twice, allowing us to reuse existing graph drawing techniques.

Our drawings are still somewhat unsatisfying for use as lace patterns.  Consider an edge $e \in E(G)$ that crosses the rectangular
schema, say the horizontal border $L$.  The edge is broken into two parts $e_t$ and $e_b$, where $e_t$ intersects the top of the rectangle along $L_t$ and $e_b$ intersects the bottom along $L_b$.
The algorithm of Aleardi~et~al.~\cite{aleardi2014}
ensures that the end points of $e_t$ and $e_b$ on $L$ have the same $x$-coordinate but
for $e$ to be truly seen as ``one edge'', $e_t$ and $e_b$ should also have
the same slope.  The bend is clearly visible when viewing several tiled repeats of the pattern as shown in
Figure~\ref{fig:bends}.
At the current time, we cannot see a way to simultaneously satisfy the topological $(1,0)$ homotopy constraint enforced by the rectangular representation and provide smooth continuity of edges across this rectangular boundary.

\vfigbegin
    \def\svgwidth{0.8\textwidth}
\begingroup%
  \makeatletter%
  \providecommand\color[2][]{%
    \errmessage{(Inkscape) Color is used for the text in Inkscape, but the package 'color.sty' is not loaded}%
    \renewcommand\color[2][]{}%
  }%
  \providecommand\transparent[1]{%
    \errmessage{(Inkscape) Transparency is used (non-zero) for the text in Inkscape, but the package 'transparent.sty' is not loaded}%
    \renewcommand\transparent[1]{}%
  }%
  \providecommand\rotatebox[2]{#2}%
  \ifx\svgwidth\undefined%
    \setlength{\unitlength}{2139.57840567bp}%
    \ifx\svgscale\undefined%
      \relax%
    \else%
      \setlength{\unitlength}{\unitlength * \real{\svgscale}}%
    \fi%
  \else%
    \setlength{\unitlength}{\svgwidth}%
  \fi%
  \global\let\svgwidth\undefined%
  \global\let\svgscale\undefined%
  \makeatother%
  \begin{picture}(1,0.2610527)%
    \put(0,0){\includegraphics[width=\unitlength,page=1]{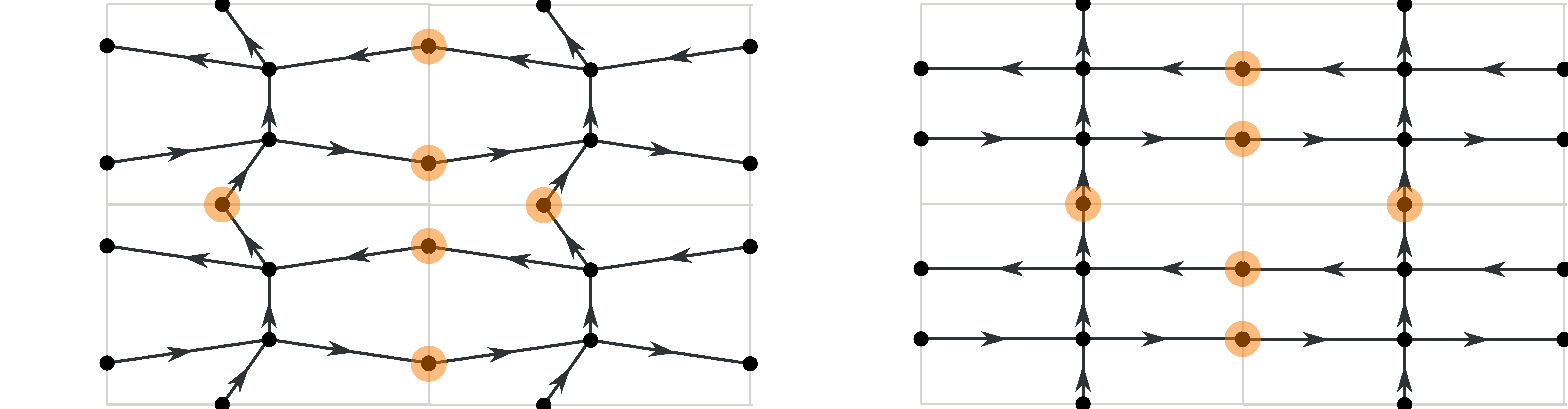}}%
    \put(-0,0.1234643){\color[rgb]{0,0,0}\makebox(0,0)[lb]{\smash{\g{(a)}}}}%
    \put(0.51917549,0.1239062){\color[rgb]{0,0,0}\makebox(0,0)[lb]{\smash{\g{(b)}}}}%
  \end{picture}%
\endgroup%

\vfigend{fig:bends}{Unwanted bends along border. (a) Result of algorithm. (b) Desired result.}{Unwanted bends}


\newpage
\bibliographystyle{splncs03}

\bibliography{lacegraph} 

\newpage
\begin{appendix}
\section*{Appendix A: Introduction to bobbin lace}

Perhaps the easiest way to understand bobbin lace is to look at the six step process by which it is made.
\vfigbegin
    \def\svgwidth{0.7\textwidth}
\begingroup%
  \makeatletter%
  \providecommand\color[2][]{%
    \errmessage{(Inkscape) Color is used for the text in Inkscape, but the package 'color.sty' is not loaded}%
    \renewcommand\color[2][]{}%
  }%
  \providecommand\transparent[1]{%
    \errmessage{(Inkscape) Transparency is used (non-zero) for the text in Inkscape, but the package 'transparent.sty' is not loaded}%
    \renewcommand\transparent[1]{}%
  }%
  \providecommand\rotatebox[2]{#2}%
  \ifx\svgwidth\undefined%
    \setlength{\unitlength}{1791.6bp}%
    \ifx\svgscale\undefined%
      \relax%
    \else%
      \setlength{\unitlength}{\unitlength * \real{\svgscale}}%
    \fi%
  \else%
    \setlength{\unitlength}{\svgwidth}%
  \fi%
  \global\let\svgwidth\undefined%
  \global\let\svgscale\undefined%
  \makeatother%
  \begin{picture}(1,0.33495925)%
    \put(0,0){\includegraphics[width=\unitlength]{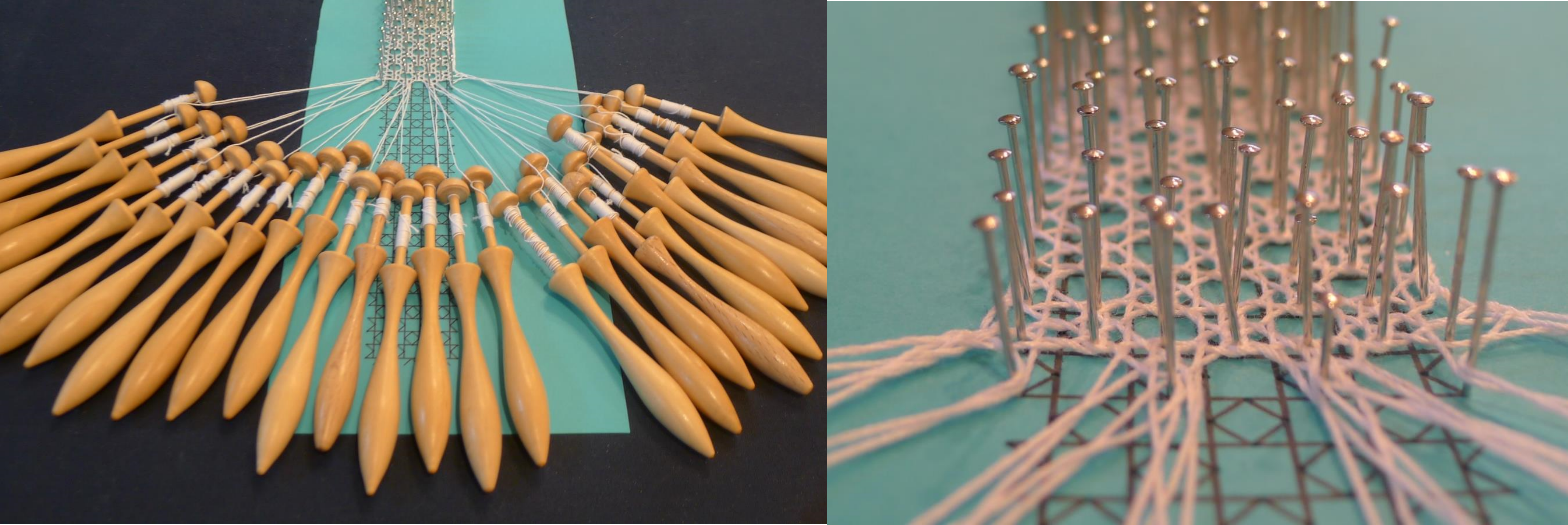}}%
  \end{picture}%
\endgroup%

\vfigend{fig:laceMaking}{Pattern in progress on lace pillow.}{Making lace}

\textbf{Step 1) Prepare the threads.}  To manage many long threads without creating a tangled mess, each end of a thread is wound onto one of a pair of bobbins.  A bobbin, commonly made from wood, is about 10cm in length.  One end is flanged to hold a length of thread and the other end, usually thicker and sometimes weighted with beads, is the handle.

\textbf{Step 2) Prepare the pattern.}  The lace is worked on top of a firm pillow stuffed with material that can hold pins (such as styrofoam or sawdust).   As threads are braided together, they are held in place by pins pushed into the pillow.  To start a piece of lace, a pattern is copied onto stiff material such as card stock.  Black dots in the pattern represent the position of pins.  Before starting to make the lace, all of these dots are pricked through to make small holes.  The pattern is then pinned to the pillow.

\textbf{Step 3) Hang the bobbins.} The middle of each thread is draped around an anchoring pin at the top of the pattern with the pair of bobbins hanging down on either side. Often the first row of the pattern is a simple weave to anchor the threads.

\textbf{Step 4) Braid and pin.} The lacemaker braids the threads working with four consecutive threads at a time.
During a sequence of braiding actions, the lacemaker may insert a pin to hold the braid in place.  The pin provides resistance so that the lacemaker can apply tension to an individual thread without distorting its neighbours.

\textbf{Step 5) Advance to next set.} Once four threads have been braided and pinned, the lace maker moves on to braid another set of four consecutive threads.  This new set of threads may include two threads from the previous set but it may also be formed from four completely new threads.  Which four threads comprise the next set depends on the pattern.

Steps 4 and 5 are repeated until the lace pattern is completed.

\textbf{Step 6) Finish.} The threads are secured (sometimes with a knot, sometimes by weaving them back into the lace) and trimmed off.  The pins are removed and the lace may be lifted off the pillow.  The finished lace is held together by the over and under crossings of the threads and friction.

\section*{Appendix B: Proof of Lemma~\ref{lem:cycleFace}}

\noindent{\bf Lemma~\ref{lem:cycleFace}.}
{\it
Presume a $2,2$-regular digraph has a toroidal cellular embedding $G$.
If $G$ has a contractible directed circuit $C$,
then $G$ will have at least one face bounded by a contractible directed circuit.
}
\vfigbegin
    \def\svgwidth{0.9\textwidth}
\begingroup%
  \makeatletter%
  \providecommand\color[2][]{%
    \errmessage{(Inkscape) Color is used for the text in Inkscape, but the package 'color.sty' is not loaded}%
    \renewcommand\color[2][]{}%
  }%
  \providecommand\transparent[1]{%
    \errmessage{(Inkscape) Transparency is used (non-zero) for the text in Inkscape, but the package 'transparent.sty' is not loaded}%
    \renewcommand\transparent[1]{}%
  }%
  \providecommand\rotatebox[2]{#2}%
  \ifx\svgwidth\undefined%
    \setlength{\unitlength}{594.98745117bp}%
    \ifx\svgscale\undefined%
      \relax%
    \else%
      \setlength{\unitlength}{\unitlength * \real{\svgscale}}%
    \fi%
  \else%
    \setlength{\unitlength}{\svgwidth}%
  \fi%
  \global\let\svgwidth\undefined%
  \global\let\svgscale\undefined%
  \makeatother%
  \begin{picture}(1,0.21361795)%
    \put(0,0){\includegraphics[width=\unitlength]{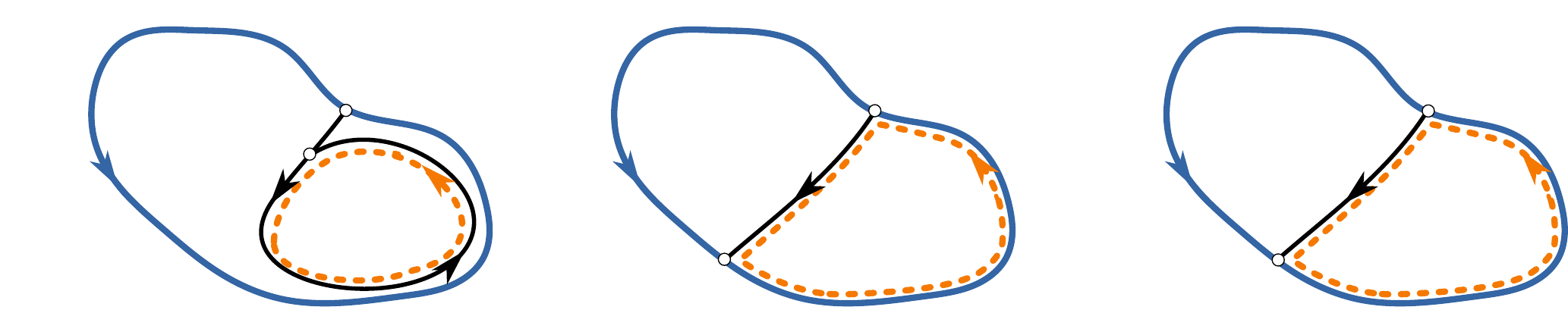}}%
    \put(0.3334524,0.13158001){\color[rgb]{0,0,0}\makebox(0,0)[lb]{\smash{\g{(b)}}}}%
    \put(0.6830396,0.13158001){\color[rgb]{0,0,0}\makebox(0,0)[lb]{\smash{\g{(c)}}}}%
    \put(0.53888638,0.06861527){\color[rgb]{0,0,0}\makebox(0,0)[lb]{\smash{$\times$}}}%
    \put(0.41996453,0.16975159){\color[rgb]{0.20392157,0.39607843,0.64313725}\makebox(0,0)[lb]{\smash{\s{$C$}}}}%
    \put(0.52769664,0.03671875){\color[rgb]{0.80784314,0.36078431,0}\makebox(0,0)[lb]{\smash{\s{$C'$}}}}%
    \put(0.55856467,0.14923442){\color[rgb]{0,0,0}\makebox(0,0)[lb]{\smash{\s{$v$}}}}%
    \put(0.43692269,0.02988702){\color[rgb]{0,0,0}\makebox(0,0)[lb]{\smash{\s{$w$}}}}%
    \put(0.78037028,0.12517174){\color[rgb]{0,0,0}\makebox(0,0)[lb]{\smash{$\times$}}}%
    \put(0.77224086,0.16975159){\color[rgb]{0.20392157,0.39607843,0.64313725}\makebox(0,0)[lb]{\smash{\s{$C$}}}}%
    \put(0.87997297,0.03671875){\color[rgb]{0.80784314,0.36078431,0}\makebox(0,0)[lb]{\smash{\s{$C'$}}}}%
    \put(0.91156199,0.14908432){\color[rgb]{0,0,0}\makebox(0,0)[lb]{\smash{\s{$v$}}}}%
    \put(0.78894772,0.02994414){\color[rgb]{0,0,0}\makebox(0,0)[lb]{\smash{\s{$w$}}}}%
    \put(0,0.13158001){\color[rgb]{0,0,0}\makebox(0,0)[lb]{\smash{\g{(a)}}}}%
    \put(0.08651213,0.16975159){\color[rgb]{0.20392157,0.39607843,0.64313725}\makebox(0,0)[lb]{\smash{\s{$C$}}}}%
    \put(0.20500077,0.04747528){\color[rgb]{0.80784314,0.36078431,0}\makebox(0,0)[lb]{\smash{\s{$C'$}}}}%
    \put(0.22123025,0.14944432){\color[rgb]{0,0,0}\makebox(0,0)[lb]{\smash{\s{$v$}}}}%
    \put(0.1656479,0.11669336){\color[rgb]{0,0,0}\makebox(0,0)[lb]{\smash{\s{$w$}}}}%
  \end{picture}%
\endgroup%

\vfigend{fig:icoc}{Three cases in proof of Lemma \ref{lem:cycleFace}.  Origin is represented by $\times$. (a) $C'$ stays inside $C$ throughout (b) Origin inside $C'$. (b) Origin outside $C'$.}{Proof contractible cycle is a face}

\begin{proof}
Arbitrarily declare one face to contain the origin so that ``inside'' and
``outside'' are well-defined for any contractible directed cycle $C$.
Let $O_C$ and $I_C$ be the number of faces outside and inside $C$,
and consider a directed contractible cycle $C$ that maximizes $|O_C - I_C|$.

Assume that $C$ is not a face and therefore there exists some vertex $v\in C$
with an incident edge $e$ that is not on $C$ and belongs to the inside of $C$.
Follow a directed path $P$ starting with $e$ while staying inside $C$ until we
reach a vertex $w$ where either (a) $P$ returns to a vertex on $P$ for the first time,
or (b) $P$ reaches another vertex on $C$ for the first time.
Because $G$ is a $2,2$-regular digraph with a finite number of vertices,
one of these two cases must eventually occur.
(Note: if $v$ is the head endpoint of $e$, $P$ follows the edges in reverse direction).

If $w\in P$, then the path from $w$ to $w$ along $P$ has formed
another directed cycle $C'$ that stays inside $C$ throughout and
therefore is also contractible (see Figure~\ref{fig:icoc}(a)).

If $w\in C$,then we can find a directed cycle $C'$ by taking $P$ from
$v$ to $w$ and then $C$ from $w$ to $v$.

Consider the case where $O_C \ge I_C$ (the other case is symmetric).
Examine $|O_{C'} - I_{C'}|$.
If the origin is inside $C'$ (see Figure~\ref{fig:icoc}(b)), then $I_{C'} < I_C$ and $O_{C'} > O_C$,
therefore
$|O_{C'} - I_{C'}| \geq O_{C'} - I_{C'} > O_C-I_C = |O_C - I_C|$,
a contradiction.
If the origin is outside $C'$ (see Figure~\ref{fig:icoc}(c)) then the outside of $C'$ is strictly inside
the inside of $C$, so $O_{C'}<I_C$ and $I_{C'}>O_C$, which implies
$|O_{C'}-I_{C'}|
=|I_{C'}-O_{C'}|
\geq I_{C'}-O_{C'} > O_C-I_C=|O_C-I_C|$, a contradiction to the choice of $C$.

We therefore conclude that no vertex on $C$ has an incident edge strictly inside $C$
and therefore the inside of $C$ is a face as required.

We can easily extend this proof to the case where $C$ is a contractible directed circuit by noting that any transverse intersection in $C$ can be replaced by an osculating intersection thus dividing $C$ into two or more smaller circuits.  We therefore assume that any time $C$ visits a vertex $v$ twice, it does not cross itself transversely.  We can very closely approximate a circuit $C$ that is free of transverse crossings by duplicating any repeated vertex $v$ and moving the two copies slightly apart.  The result is a simple directed cycle.
\qed\end{proof}

\section*{Appendix C: Proof of Lemma~\ref{lem:connectingPath}}

Before we can prove Lemma~\ref{lem:connectingPath}, we need a few definitions and observations.
We will call a simple directed circuit $C$
a {\em leftmost circuit} if all edges of $C$ are left outgoing at their
respective tails.  Observe that any leftmost circuit $C$ in a $2$-$2$-regular digraph is actually a cycle,
because no vertex can have two left outgoing edges incident to it.

Consider two simple cycles $P,Q$.  We already defined the algebraic
crossing number $\hat{i}(P,Q)$ if $P$ and $Q$ meet in a finite number
of points.  However, we need a similar concept even if $P$ and $Q$
share one or more edges.  To define this, replace $P$ and $Q$ by cycles $P'$ and $Q'$
that are arbitrarily close to $P$ and $Q$ in such a way that $P'$ and $Q'$
intersect in a finite number of points, and define
$\hat{i}(P,Q):=\hat{i}(P',Q')$.  It is not hard to verify that this number
is independent of the choice of $P'$ and $Q'$ since we use the algebraic
crossing number (rather than the geometric one).

Now we can make the following observation of leftmost cycles:

\begin{lemma}
Presume (C1-C3) hold.  Let $C$ be a leftmost cycle.  Then $C$ does not bound a face, and
$\hat{i}(C,P)\neq 0$ for {\em all} osculating circuits $P\in \Part$.
\label{lem:leftmost1}
\end{lemma}
\begin{proof}
$C$ cannot possibly bound a face, because it is a directed circuit
(hence non-contractible by (C3)) while  facial circuits are contractible
by (C2).  Therefore, there must be some edge $e_\ell$ on the left side of $C$ that shares a vertex $v$ with $C$, but does not  belong to $C$.
Since $C$ uses only left outgoing edges, $e_\ell$ is necessarily a
left incoming edge.  Let $P_\ell$ be the osculating circuit containing
$e_\ell$.  This circuit continues from $e_\ell$ along $C$ for possibly
some time, but eventually must leave $C$ again to return back to $e_\ell$.
$P_\ell$ can leave $C$ only at a right outgoing edge, because the left
outgoing edges always belong to $C$.   Therefore in this stretch $P_\ell$
has crossed $C$ going from left to right.

Note that $P_\ell$ can never cross $C$ from right to left, because
all left outgoing edges incident to a vertex in $C$ belong to $C$.  Therefore
$\hat{i}(P_\ell,C)\neq 0$.

Following from Lemma~\ref{lem:homClass},  $P_\ell$ and $C$ do not belong to the same
homotopy class.  From Lemma \ref{lem:pHomClass}, we know that all $P\in\Part$ belong to the same homotopy class so we can conclude $\hat{i}(P,C)\neq 0$ for all $P\in\Part$.
\qed\end{proof}

Now we can prove our main theorem:

\noindent{\bf Lemma~\ref{lem:connectingPath}.}
{\it
Presume (C1-C3) hold.  Let $P$ be a simple osculating cycle.  Then there exists a directed walk $W$
in $G$ that starts at a vertex $v\in P$ with a left outgoing edge
$e_1\not\in P$, ends at a vertex $w\in P$ with a right incoming edge
$e_k\not\in P$, and has no transverse intersection or shared edges with $P$.
}
\begin{proof}
Since $P$ is a directed cycle, it is non-contractible  by (C3) and does not bound a
face by (C2).  Therefore, there must be some edge on the left side of $P$ that shares a vertex $v$ with $P$, but does not  belong
to $P$.  If, say, the left incoming edge of $v$ is not in $P$,
then $P$ uses the right incoming and right outgoing edge at $v$ (because
$P$ follows rotationally consecutive edges), and so the left outgoing edge of $v$ is also not in $P$.
Therefore edge $e_1$ exists.

Now create a leftmost walk $W_\ell$ starting with $e_1$, then taking the left
outgoing edge at its head, and continue taking left outgoing edges until
one of the following happens:
\begin{enumerate}
\item We reach a vertex $w\in P$ via a right incoming edge.  In this case
we have found the required path $W$ and stop.
Since we stop as soon as the head vertex $w$ of a right incoming edge $e_k$ of $W_\ell$ is also a vertex of $P$, we can conclude that, before arriving at $w$, if we visited any other vertices in $P$, then these vertices were
reached from a left incoming edge, $e_i, 1\le i < k$.  Since $W_\ell$ always follows the left outgoing edge and $P$ always follows rotationally consecutive edges, we can conclude that
the left incoming edge $e_i$ and left outgoing edge $e_{i+1}$ were not edges in $P$ and the meeting of $W_\ell$ and $P$ was an osculating intersection rather than a transverse intersection.

\item $W_\ell$ repeats a vertex, and the resulting directed cycle $C$ does
not contain vertices of $P$.  Then $\hat{i}(C,P)=0$, which contradicts
Lemma~\ref{lem:leftmost1}.  This case cannot happen.

\item $W_\ell$ repeats a vertex, and the resulting directed cycle $C$
contains vertices of $P$.  Since we did not stop with Case 1, all
such vertices of $C$ reach $P$ from the left.  Following the same argument used
in Case 1, all meetings of $C$ with $P$ are osculating resulting in $\hat{i}(C,P)=0$.  This again contradicts
Lemma~\ref{lem:leftmost1} so this case cannot happen.
\end{enumerate}
The graph is finite and neither Case 2 nor 3 can happen, therefore
Case 1 eventually must happen and we have found $W$.
\qed\end{proof}

\section*{Appendix D: An example of the meridian and longitude}

\vfigbegin
    \def\svgwidth{\textwidth}
\begingroup%
  \makeatletter%
  \providecommand\color[2][]{%
    \errmessage{(Inkscape) Color is used for the text in Inkscape, but the package 'color.sty' is not loaded}%
    \renewcommand\color[2][]{}%
  }%
  \providecommand\transparent[1]{%
    \errmessage{(Inkscape) Transparency is used (non-zero) for the text in Inkscape, but the package 'transparent.sty' is not loaded}%
    \renewcommand\transparent[1]{}%
  }%
  \providecommand\rotatebox[2]{#2}%
  \ifx\svgwidth\undefined%
    \setlength{\unitlength}{982.60968018bp}%
    \ifx\svgscale\undefined%
      \relax%
    \else%
      \setlength{\unitlength}{\unitlength * \real{\svgscale}}%
    \fi%
  \else%
    \setlength{\unitlength}{\svgwidth}%
  \fi%
  \global\let\svgwidth\undefined%
  \global\let\svgscale\undefined%
  \makeatother%
  \begin{picture}(1,0.66763452)%
    \put(0,0){\includegraphics[width=\unitlength,page=1]{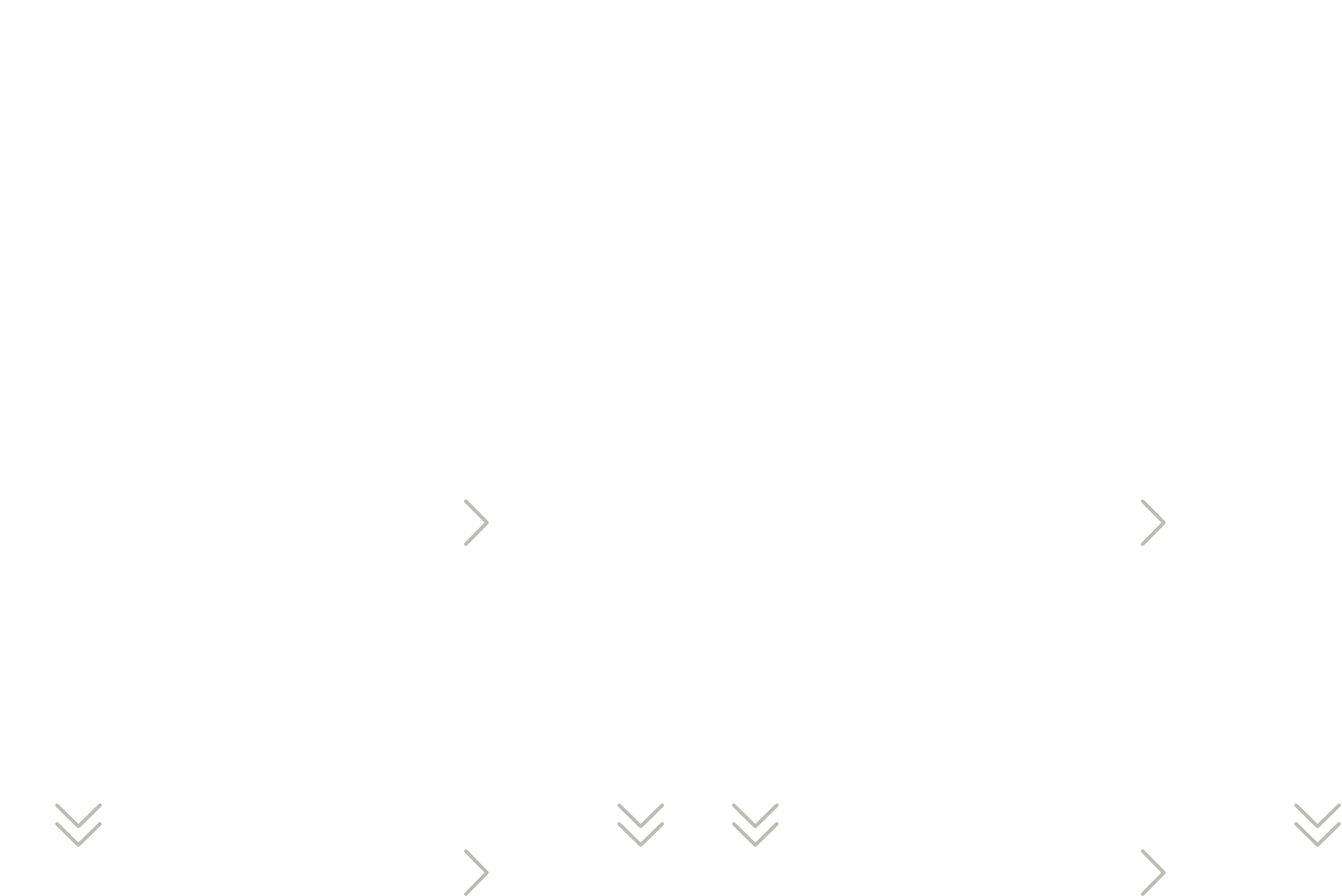}}%
    \put(0.01389396,0.47516779){\color[rgb]{0.05490196,0.1372549,0.18039216}\makebox(0,0)[b]{\smash{\g{(a)}}}}%
    \put(0.52681385,0.47516779){\color[rgb]{0.05490196,0.1372549,0.18039216}\makebox(0,0)[b]{\smash{\g{(b)}}}}%
    \put(0.01389396,0.14612557){\color[rgb]{0.05490196,0.1372549,0.18039216}\makebox(0,0)[b]{\smash{\g{(c)}}}}%
    \put(0.52681385,0.14612557){\color[rgb]{0.05490196,0.1372549,0.18039216}\makebox(0,0)[b]{\smash{\g{(d)}}}}%
    \put(0,0){\includegraphics[width=\unitlength,page=2]{find_l_m_2b.pdf}}%
    \put(0.07985699,0.51417764){\color[rgb]{0.05490196,0.1372549,0.18039216}\makebox(0,0)[lb]{\smash{\s{0}}}}%
    \put(0.14837167,0.51417764){\color[rgb]{0.05490196,0.1372549,0.18039216}\makebox(0,0)[lb]{\smash{\s{1}}}}%
    \put(0.2058091,0.51417764){\color[rgb]{0.05490196,0.1372549,0.18039216}\makebox(0,0)[lb]{\smash{\s{1}}}}%
    \put(0.26105018,0.51417764){\color[rgb]{0.05490196,0.1372549,0.18039216}\makebox(0,0)[lb]{\smash{\s{0}}}}%
    \put(0.3182165,0.51417764){\color[rgb]{0.05490196,0.1372549,0.18039216}\makebox(0,0)[lb]{\smash{\s{0}}}}%
    \put(0.58547509,0.51417764){\color[rgb]{0.05490196,0.1372549,0.18039216}\makebox(0,0)[lb]{\smash{\s{0}}}}%
    \put(0.6455266,0.46947057){\color[rgb]{0.05490196,0.1372549,0.18039216}\makebox(0,0)[lb]{\smash{\s{0}}}}%
    \put(0.64584465,0.42556625){\color[rgb]{0.05490196,0.1372549,0.18039216}\makebox(0,0)[lb]{\smash{\s{1}}}}%
    \put(0.64584465,0.39176036){\color[rgb]{0.05490196,0.1372549,0.18039216}\makebox(0,0)[lb]{\smash{\s{1}}}}%
    \put(0.66128608,0.37920583){\color[rgb]{0.05490196,0.1372549,0.18039216}\makebox(0,0)[lb]{\smash{\s{1}}}}%
    \put(0.70438728,0.37920583){\color[rgb]{0.05490196,0.1372549,0.18039216}\makebox(0,0)[lb]{\smash{\s{1}}}}%
    \put(0.73143299,0.37924161){\color[rgb]{0.05490196,0.1372549,0.18039216}\makebox(0,0)[lb]{\smash{\s{0}}}}%
    \put(0.76727405,0.33180386){\color[rgb]{0.05490196,0.1372549,0.18039216}\makebox(0,0)[lb]{\smash{\s{0}}}}%
    \put(0.79475562,0.58837856){\color[rgb]{0.05490196,0.1372549,0.18039216}\makebox(0,0)[lb]{\smash{\s{1}}}}%
    \put(0.90851907,0.51663084){\color[rgb]{0.05490196,0.1372549,0.18039216}\makebox(0,0)[lb]{\smash{\s{1}}}}%
    \put(0,0){\includegraphics[width=\unitlength,page=3]{find_l_m_2b.pdf}}%
  \end{picture}%
\endgroup%

\vfigend{fig:find_l_m_2}{Find rectangular representation. (a) Shortest path $L$. (b) Alternate $L$. (c) Draw graph with polygonal schema defined in \emph{(a)}. (d)Draw graph with polygonal schema defined in \emph{(b)}.}{Find rectangular representation}

Figure~\ref{fig:find_l_m_2} illustrates how the graph $G$ (of Figure~\ref{fig:find_l_m}(b))
is converted into a rectangular representation.  We consider two different
choices for $L$.  Note that while both of them are perfectly fine longitudes,
the left one is more visually appealing, principally because $L$ and $G$ have significantly fewer crossings.  To minimize the number of times $L$ crosses $G$, assign a weight to the edges of $\mathcal{L}(G)$ according to the number of endpoints the edge has that are also a vertex or dummy vertex of an edge in $G$.

\end{appendix}

\end{document}